\documentclass{llncs}

\usepackage{amssymb}
\usepackage{amsmath}
\usepackage[hyphens]{url}
\usepackage{hyperref}
\usepackage{graphicx}
\usepackage{subfigure}
\usepackage{xspace}
\usepackage{pifont}
\usepackage{changepage}

\spnewtheorem{metatheorem}[theorem]{Metatheorem}{\bfseries}{\itshape}

\newcommand{\cmark}{\ding{51}}
\newcommand{\complexityclass}[1]{\textbf{#1}\xspace}
\newcommand{\computproblem}[1]{\textsc{#1}\xspace}
\renewcommand{\P}{\complexityclass{P}}
\newcommand{\NP}{\complexityclass{NP}}

\newcommand{\APX}{\complexityclass{APX}}

\newcommand{\PSPACE}{\complexityclass{PSPACE}}
\newcommand{\NPSPACE}{\complexityclass{NPSPACE}}

\newcommand{\MTSAT}{\computproblem{Max-3-Sat}}
\newcommand{\LEM}{\computproblem{Lemmings}}
\newcommand{\CLI}{\computproblem{Climbers}}
\newcommand{\FLO}{\computproblem{Floaters}}
\newcommand{\BLO}{\computproblem{Blockers}}
\newcommand{\BOM}{\computproblem{Bombers}}
\newcommand{\BUI}{\computproblem{Builders}}
\newcommand{\BAS}{\computproblem{Bashers}}
\newcommand{\MIN}{\computproblem{Miners}}
\newcommand{\DIG}{\computproblem{Diggers}}
\newcommand{\WGW}{\computproblem{Weighted Graph Walk}}
\newcommand{\QBF}{\computproblem{Quantified Boolean Formula}}

\begin{document}
\mainmatter

\title{Lemmings Is \PSPACE-Complete}

\titlerunning{Lemmings Is \PSPACE-Complete}

\author{Giovanni Viglietta}

\authorrunning{Giovanni Viglietta}

\tocauthor{Giovanni Viglietta}

\institute{University of Ottawa, Canada,\\
\email{viglietta@gmail.com}}

\maketitle

\begin{abstract}
Lemmings is a computer puzzle game developed by DMA Design and published by Psygnosis in 1991, in which the player has to guide a tribe of lemming creatures to safety through a hazardous landscape, by assigning them specific skills that modify their behavior in different ways. In this paper we study the optimization problem of saving the highest number of lemmings in a given landscape with a given number of available skills.

We prove that the game is \PSPACE-complete, even if there is only one lemming to save, and only Builder and Basher skills are available. We thereby settle an open problem posed by Cormode in 2004, and again by Fori\v sek in 2010. However we also prove that, if we restrict the game to levels in which the available Builder skills are only polynomially many (and there is any number of other skills), then the game is solvable in {\NP}\@. Similarly, if the available Basher, Miner, and Digger skills are polynomially many, the game is solvable in \NP.

Furthermore, we show that saving the maximum number of lemmings is \APX-hard, even when only one type of skill is available, whatever this skill is. This contrasts with the membership in \P of the decision problem restricted to levels with no ``deadly areas'' (such as water or traps) and only Climber and Floater skills, as previously established by Cormode.
\end{abstract}

\section{Introduction}\label{s1}
Lemmings is a popular computer game originally developed by DMA Design for PC and Commodore Amiga. Since its first release in 1991, by Psygnosis, several ports, sequels and imitations have appeared, for various systems. The game revolves around the behavior of some creatures called \emph{lemmings}, which deterministically walk across a landscape, turning around at walls, and blindly falling into pitfalls or drowning in water. The player's goal is to guide the highest number of lemmings through the landscape, from their respective \emph{entrance locations} to any \emph{exit location}, within a certain amount of time. To do so, the player has an arsenal of \emph{skills} that they can individually assign to lemmings, in order to modify their behavior in different ways, and hopefully prevent them from perishing. Because the number of available skills is limited, and most skills have just a temporary effect, the player must carefully plan their strategy, which makes Lemmings a challenging puzzle game.

In this paper we study the computational complexity of the optimization problem of saving the highest number of lemmings in a given level of the game, contributing to a fast-growing branch of research delightfully surveyed in \cite{survey2,survey1}.

In~\cite{mccarthy}, McCarthy first studied the game of Lemmings as an archetypical model for the logical approach to AI, attempting a formalization of the game using situation calculus, and discussing the features that make Lemmings a challenge to both experimental and theoretical AI\@. Spoerer later used genetic algorithms to generate successful solutions for a severely simplified version of Lemmings~\cite{spoerer}.

In~\cite{cormode}, Cormode established several complexity results related to another simplified version of Lemmings. In Cormode's model, the landscape contains no \emph{deadly areas} such as water, lava or traps, the player can assign skills to several different lemmings at the same time instant, and the time limit to complete each level is bounded by a polynomial in the size of the level itself. Cormode's paper shows the \NP-completeness of deciding if a level of such a game is solvable, even when only a single lemming is present, and only Digger skills are available. It is further shown that, if only Floater and Climber skills are available, then solvability is decidable in \P.

The rationale behind Cormode's assumption on the time limit is the claim that any level of Lemmings is either unsolvable or solvable within a polynomial amount of time. Later, in~\cite{forisek}, Fori\v sek disproved such a claim by constructing a class of levels whose solutions involve ``waiting'' an exponentially long time for certain configurations to occur, hence suggesting that the full Lemmings game may fail to be in {\NP}\@. Both Cormode and Fori\v sek conjectured that Lemmings, with no restrictions, is \PSPACE-complete. Cormode also asked for the computational complexity of classes of game instances with different combinations of initially available skills.

More recently, in~\cite{viglietta}, the author gave an independent \NP-hardness proof that works for instances with only Basher skills, and observed that a similar argument can be extended to instances with only Miner skills.

\begin{table}
\begin{center}
{\renewcommand{\arraystretch}{1.25}
\begin{tabular}{c||c|c|c||c|c|c}
Climbers & $\infty$ & poly & $\infty$ & & & \\
\hline
Floaters & $\infty$ & poly & $\infty$ & & & \\
\hline
Bombers & & poly & $\infty$ & & & \\
\hline
Blockers & & poly & $\infty$ & & & \\
\hline
Builders & & poly & $\infty$ & & & \\
\hline
Bashers & & poly & $\infty$ & & & poly \\
\hline
Miners & & poly & $\infty$ & & poly & \\
\hline
Diggers & & poly & $\infty$ & poly & & \\
\hline\hline
Time & $\infty$ & poly & $\infty$ & poly & poly & poly \\
\hline
Hazards & & & \cmark & & & \\
\hline
1 lemming & & & & \cmark & & \\
\hline\hline
& \P & \NP & \PSPACE & \NP-hard & \NP-hard & \NP-hard
\end{tabular}}
\end{center}
\caption{Previously known results}
\label{tab1}
\end{table}

Table~\ref{tab1} summarizes the aforementioned results. Note that the membership of Lemmings in \PSPACE is folklore. (A formal proof of this fact will be given in Theorem~\ref{t2}, as part of the proof that the game is \PSPACE-complete.)

\paragraph{\textbf{\emph{Our contribution.}}} In Section~\ref{s2} we define \LEM, the optimization problem of maximizing the number of saved lemmings in a given level. One of the novelties of our approach is that we do not aim at studying a simplified or conveniently modified version of the game, but our model incorporates every aspect and feature of the original Lemmings game developed by DMA Design, including the known glitches.\footnote{Refer to \url{http://www.lemmingsforums.com/index.php?topic=525.0}.}
(The only, obvious, exception is that we allow arbitrarily large levels with arbitrarily many \emph{objects}.)

In Section~\ref{s3} and Section~\ref{s4} we argue that what separates the ``harder'' levels of \LEM from the ``easier'' ones is the number of constructive and destructive commands that can be assigned to lemmings. Namely, if the number of Builder skills and the number of Basher skills are both exponential in the size of the level (or \emph{unlimited}), then we are able to construct a \PSPACE-complete class of instances with the bonus feature of having only one lemming each (Theorem~\ref{t2}). Conversely, we show that the decision problem restricted to instances with only polynomially many available Builder skills (and any number of other skills) belongs to \NP (Theorem~\ref{t1}). Similarly, if polynomially many Basher, Miner, and Digger skills (and any number of other skills) are available, the decision version belongs to \NP (Theorem~\ref{t1b}). We thus provide an adequate answer to the open problem of Cormode and Fori\v sek on the complexity of the full Lemmings game.

\begin{table}
\begin{adjustwidth}{-1.75cm}{}
{\renewcommand{\arraystretch}{1.25}
\begin{tabular}{c||c|c||c||c|c|c|c|c|c|c|c}
Climbers & $\infty$ & $\infty$ & & poly & & & & & & & \\
\hline
Floaters & $\infty$ & $\infty$ & & & poly & & & & & & \\
\hline
Bombers & $\infty$ & $\infty$ & & & & poly & & & & & \\
\hline
Blockers & $\infty$ & $\infty$ & & & & & poly & & & & \\
\hline
Builders & poly & $\infty$ & exp & & & & & poly & & & \\
\hline
Bashers & $\infty$ & poly & exp & & & & & & poly & & \\
\hline
Miners & $\infty$ & poly & & & & & & & & poly & \\
\hline
Diggers & $\infty$ & poly & & & & & & & & & poly \\
\hline\hline
Time & $\infty$ & $\infty$ & exp & poly & poly & poly & poly & poly & poly & poly & poly \\
\hline
Hazards & \cmark & \cmark & \cmark & \cmark & \cmark & \cmark & \cmark & \cmark & \cmark & \cmark & \cmark \\
\hline
1 lemming & & & \cmark & & & & & & & & \\
\hline\hline
& \NP & \NP & \PSPACE-h. & \APX-h. & \APX-h. & \APX-h. & \APX-h. & \APX-h. & \APX-h. & \APX-h. & \APX-h.\\
\end{tabular}}
\\\\
\caption{New results}
\label{tab2}
\end{adjustwidth}
\end{table}

In Section~\ref{s5} we discuss the restriction of \LEM to instances with only one type of available skill, and we give a proof of \APX-hardness of all such sub-games (there are eight in total, one for each skill type). This also implies that computing approximate solutions with a relative error lower than $1/8$ is \NP-hard (Corollary~\ref{cor1}). Combined with Cormode's results, this suggests that what makes levels with only Climber and Floater skills ``hard'' is the presence of \emph{traps}.

Section~\ref{s6} concludes the paper with some final remarks and open problems.

The main results of this paper are summarized in Table~\ref{tab2}. All our constructions have been tested with the DOS version of the original Lemmings game, and can be downloaded as a \emph{level pack} from \url{http://giovanniviglietta.com/files/lemmings/gadgets.dat}.

A conference version of this paper has appeared at FUN'14~\cite{viglietta2}.

\section{Game Definition}\label{s2}
We model \LEM as an optimization problem (refer to~\cite{ausiello}) whose instances are \emph{levels} of the form $\mathcal L=(time, terrain, steel, objects, lemmings, rate, skills)$.

\paragraph{\textbf{\emph{Time.}}} In Lemmings, time is discretized and subdivided into \emph{time units}. Accordingly, in each level of \LEM, \emph{time} is the amount of time units that the player has to achieve the goal of saving as many lemmings as possible. The value of \emph{time} is assumed to be at most exponential in the size of the landscape (see below), or \emph{unlimited}.

\paragraph{\textbf{\emph{Landscape.}}} \emph{terrain}, \emph{steel} and \emph{objects} collectively define the \emph{landscape} of the level:
\begin{itemize}
\item[$\bullet$] \emph{terrain} is a rectangular array of \emph{cells}, each of which is the size of a pixel and can be \emph{empty} or \emph{solid}. Informally, this is a bitmap containing the ``shape'' of the landscape: lemmings can freely walk across empty cells, but are stopped by solid cells. It is convenient to consider \emph{terrain} as (logically) partitioned into \emph{blocks} of $4\times 4$ cells.
\item[$\bullet$] \emph{steel} is a rectangular array that tells whether each \emph{terrain} block is ``made of steel'' or is ``permeable''. It may be viewed as a block-aligned ``mask'' that is overlaid on \emph{terrain}, and is used to check if solid \emph{terrain} cells may be ``excavated'' by Bombers, Bashers, Miners or Diggers (see below). Notice that each \emph{terrain} $4\times 4$ block is either entirely made of steel or entirely permeable, regardless of the amount of solid cells that it actually contains. In particular, even a block consisting of empty cells can be made of steel.
\item[$\bullet$] \emph{objects} is an array (of length polynomial in the size of \emph{terrain}) whose elements have a \emph{position} within the landscape, a \emph{trigger area}, a \emph{type}, and an optional \emph{delay} parameter (whose value is bounded by a polynomial). Like steel masks, trigger areas are block-aligned bitmaps that are overlaid on \emph{terrain}. However, if a block hosts the trigger area of some object, it cannot be made of steel, and hence it must be permeable. There are four types of objects:
\begin{itemize}
\item[$\circ$] \emph{Entrance}. Each lemming enters the level through an entrance. There may be several entrances in the same level (see below).
\item[$\circ$] \emph{Exit}. A lemming reaching the trigger area of an exit is ``saved'' and is removed from the game. There may be several exits in the same level.
\item[$\circ$] \emph{Deadly zone}. A lemming lying in the trigger area of a deadly zone instantly dies and is removed from the game. However, after a deadly zone has killed a lemming, it remains harmless for $k$ time units, where $k$ is the object's delay parameter, and then it becomes deadly again. During that window of $k$ time units, lemmings can safely traverse the trigger area. (Even if several lemmings enter the trigger area at the same time unit, only one is killed immediately.) Deadly zones with $k>0$ are represented in Lemmings as ``traps'', such as presses, gallows poles and electrocuting devices; deadly zones with $k=0$ are represented as water or lava, and they instantly kill every lemming that enters them.
\item[$\circ$] \emph{One-way wall}. The \emph{terrain} cells underlying its trigger area are perceived as permeable by Bashers and Miners going in one direction, and made of steel by Bashers and Miners going in the opposite direction (see below).
\end{itemize}
\end{itemize}

Notice that both \emph{steel} and the trigger areas of objects have a coarser resolution than \emph{terrain}, due to the file format that Lemmings uses to store levels. We also stress that a $4\times 4$ block cannot simultaneously be made of steel and be part of the trigger area of an object. These features will add an extra challenge to the design of our gadgets.

\paragraph{\textbf{\emph{Lemmings.}}} The \emph{lemmings} parameter of a level is the total amount of lemmings that the level contains (which is assumed to be bounded by a polynomial in the size of \emph{terrain}). Lemmings enter the game one at a time, at a frequency given by the parameter \emph{rate}. If several entrances are present, they release lemmings in turns, following an order determined by their position in the \emph{objects} array.

Upon entering the land, each lemming is facing right, and is normally a \emph{Faller}, which falls vertically through empty \emph{terrain} cells at constant speed, until it lands on a solid cell. Then it becomes a \emph{Walker}, which keeps walking straight (in the direction it is facing) as long as it can. In Lemmings, the sprite of a lemming is between nine and ten pixels high, depending on the animation frame. However, only one cell matters for collision detection with the landscape, which is the lemming's \emph{pin}. The pin is located one cell below the lemming's feet, and its exact position varies depending on the animation frame and the direction the lemming is facing.

On flat ground, a Walker's pin moves forward by eight cells every four time units. Between time units, the collision detection algorithm first moves the Walker's pin forward by one cell, no matter if it is solid or empty. Then \emph{terrain} cells are checked to determine the lemming's behavior.

If the pin has reached a solid cell, then the cells above are also checked. If the lowest empty cell is eight cells above the pin or higher, then the slope is too high and the Walker reverses its direction. Otherwise, the pin ``jumps'' above by at most two cells, and then goes further up by one cell per time unit, until the top is reached.

Otherwise, if the pin has reached an empty cell, the lemming falls down until it reaches a solid cell again. On the first time unit, the pin's position instantly drops by at most four cells. If a solid cell has not been reached yet, then the lemming becomes a Faller and its pin gradually moves down, by at most two cells per time unit. If the fall is longer than 63 cells (or crosses the bottom of the terrain), the lemming dies and is removed from the game.

Depending on the Walker's animation frame, the above procedure may be repeated between one and three times per time unit (on ``almost flat'' ground). Figure~\ref{f1a} illustrates an example, in which dots represent the final positions of the lemming's pin each time the collision detection algorithm is executed.

\begin{figure}[h]
\centering
\subfigure[Walker]{\label{f1a}\includegraphics[scale=1]{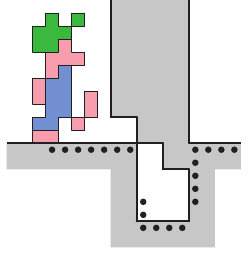}}\qquad \quad
\subfigure[Builder]{\label{f1b}\includegraphics[scale=1]{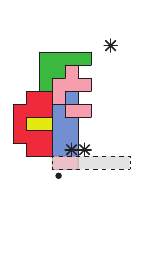}}\qquad \quad
\subfigure[Basher]{\label{f1c}\includegraphics[scale=1]{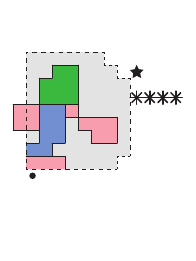}}
\caption{(a) Sequence of pins (black dots) of the lemming, as it walks rightward over the gray solid cells. (b) First brick of the stairway (dashed area), and the three cells that are tested for solidity (asterisks). (c) Cells dug on the first stroke (dashed area), the cell tested for permeability (star), and the four cells tested for solidity (asterisks).}
\label{f1}
\end{figure}

\paragraph{\textbf{\emph{Skills.}}} Finally, the level parameter \emph{skills} is an array containing the number of skills that the player can assign to lemmings. We will assume that all skill quantities are bounded by an exponential in the size of \emph{terrain}, or \emph{unlimited}. The skills are:
\begin{itemize}
\item[$\bullet$] \emph{Climber}. A permanent skill that makes a lemming climb vertical rows of more than six solid cells, at an average speed of one cell every two time units, instead of turning around like a Walker. As soon as a Climber reaches the ``top of a wall'', it starts behaving like a Walker again. If it hits a ``ceiling'' while it is climbing, it turns around and falls back down.
\item[$\bullet$] \emph{Floater}. A permanent skill that makes a lemming survive falls of any height. Floaters also fall slower than Fallers.
\item[$\bullet$] \emph{Bomber}. Makes a lemming explode after a small amount of time units. A Bomber keeps behaving normally until it actually blows up, also turning the surrounding \emph{terrain} cells from solid to empty, provided that the Bomber's pin lies on a permeable cell.
\item[$\bullet$] \emph{Blocker}. Makes a lemming stand still, and makes any Walkers, Builders, and Miners that come in contact with it turn around. A Blocker can be interrupted only by killing it (for instance by making it become a Bomber), or by excavating the solid \emph{terrain} cell containing its pin.
\item[$\bullet$] \emph{Builder}. Makes a lemming construct a ``stairway'' by turning empty \emph{terrain} cells into solid ones. Each ``brick'' of the stairway is six cells wide and one cell high, and is laid on top of the Builder's pin, as Figure~\ref{f1b} indicates. Then the Builder's pin is moved one cell up and two cells forward, and subsequently three \emph{test cells} are checked for solidity (these cells are represented as asterisks in the figure). If all the test cells are empty, the Builder lays a new brick. Otherwise, the Builder turns around and becomes a Walker. After laying 12 bricks, a Builder becomes a Walker anyway, and proceeds forward as usual.
\item[$\bullet$] \emph{Basher}. Makes a lemming ``dig'' a horizontal hole in the direction it is facing, by turning solid \emph{terrain} cells into empty cells. Upon assignment of the skill, the lemming checks the \emph{test cell} marked by a star in Figure~\ref{f1c}. If it is permeable (no matter if it is solid or empty), the lemming becomes a Basher and makes a hole shaped like the dashed area. It then proceeds digging forward at five cells per \emph{stroke}. It only stops when it falls into a hole (then it becomes a Faller), or when it encounters a steel cell in the location marked by the star (then it turns around and becomes a Walker), or when all the four cells marked by an asterisk are empty (then it becomes a Walker without turning around). One-way walls are treated as steel or permeable cells, depending on their orientation.
\item[$\bullet$] \emph{Miner}. Similar to Basher, but a Miner digs diagonally.
\item[$\bullet$] \emph{Digger}. Similar to Basher, but a Digger digs vertically.
\end{itemize}
Builders, Bashers, Miners, and Diggers can be interrupted by the player at any time by assigning them a different skill (that is not a Climber or a Floater skill). Blockers can be interrupted only by digging the solid \emph{terrain} cell on which they stand, or by assigning them a Bomber skill, which kills them.

We remark that Builder, Basher, and Miner skills can also be assigned ``ineffectively'', i.e., they may change a lemming's behavior without modifying \emph{terrain} cells. For instance, a Builder's first brick may end up on already solid cells. In this case, the cells remain solid, the lemming climbs on the brick, and it may immediately re-become a Walker if some of its test cells are solid. On the other hand, a Basher (and analogously a Miner) that is not facing solid or steel cells will stroke once and advance by a few cells without modifying \emph{terrain}. Then, if it still detects no solid cells ahead, it becomes a Walker. Otherwise, if it detects solid cells ahead after stroking once ineffectively, but it detects steel in the new steel test cell, it turns around and becomes a Walker. All these behaviors have the only net effect of delaying a Walker by a few time units, and possibly reversing its direction.

\paragraph{\textbf{\emph{Actions.}}} A player's \emph{action} is the assignment of a certain skill to a certain lemming at a certain time. In the computer game, actions are done by ``clicking'' on lemmings. At most one skill can be assigned per time unit. In particular, if several lemmings lie under the ``cursor'' at the same time, the skill is assigned only to one lemming. An action is encoded by a lemming's position in the lemmings array, a skill identifier, and a \emph{timestamp}.

A \emph{feasible solution} of \LEM is then a finite sequence of actions that are compatible with each other and with the given amount of available skills. To complete the definition of \LEM, we still need a \emph{measure function}, which is obviously the number of lemmings that the player saves within the time limit, by doing a certain sequence of actions given by a feasible solution.

\section{Instances Solvable in \NP}\label{s3}
In this section we study two restrictions of \LEM, and we prove that their decision versions are solvable in {\NP}\@. In the first restriction, the number of initially available Builder skills is bounded by a polynomial in the size of the landscape (Theorem~\ref{t1}). In the second restriction, the number of Basher, Miner, and Digger skills is polynomial (Theorem~\ref{t1b}).

These results extend~\cite[Lemma~1]{cormode} by Cormode, which states that \LEM is in \NP, provided that the time limit to solve a level is polynomial (and in particular there are polynomially many available skills).

We first introduce a subsidiary problem. A \emph{weighted directed multigraph} is a triplet of the form $G=(V,E,\mu)$, where $(V,E)$ is a directed multigraph and $\mu\colon E\to\mathbb N^m$ is a function assigning a \emph{weight} to each edge. The \WGW problem on input $(G,s,t,k)$ asks if the weighted directed multigraph $G=(V,E,\mu)$ has a walk from vertex $s\in V$ to vertex $t\in V$ whose \emph{total weight} is exactly $k\in\mathbb N^m$. The total weight of a walk is the sum of the weights of the edges that it traverses, taken with their multiplicity. (The sum in $\mathbb N^m$ is defined as $(a_1,a_2,\cdots,a_m)+(b_1,b_2,\cdots,b_m)=(a_1+b_1,a_2+b_2,\cdots,a_m+b_m)$.)

In the next lemma we prove that \WGW is solvable in {\NP}\@. Observe that a straightforward encoding of a walk (e.g., as an ordered list of edges) may not be a valid certificate for \WGW, because the length of such a walk is not necessarily polynomial in the size of the multigraph.

\begin{lemma}\label{lwgw}
$\WGW\in\NP$.
\end{lemma}
\begin{proof}
Let $(G,s,t,k)$ be an instance of \WGW on $G=(V,E,\mu)$. A certificate for \WGW is a function $\sigma\colon E\to\mathbb N$ whose encoding has size at most quadratic in the size of $(G,s,t,k)$. Intuitively, in a positive instance of \WGW, $\sigma(e)$ indicates the number of times edge $e\in E$ is traversed by a walk of total weight $k$ from $s$ to $t$. We make the verifier accept if and only if $\sigma$ passes the following tests. First, the edges of $G$ on which $\sigma$ is non-zero should induce a weakly connected subgraph of $G$ (which is easily verified in polynomial time). Second, $\sigma$ must satisfy $\sum_{e\in E}\sigma(e)\cdot\mu(e)=k$. Finally, the following equations must be satisfied: if $s=t$, then
\begin{equation}\label{eq:1}\sum_{(v,u)\in E}\sigma((v,u))-\sum_{(u,v)\in E}\sigma((u,v))=0 \quad \mbox{ for every } v\in V;\end{equation}
otherwise, if $s\neq t$,
\begin{equation}\label{eq:2}\sum_{(v,u)\in E}\sigma((v,u))-\sum_{(u,v)\in E}\sigma((u,v))=0 \quad \mbox{ for every } v\in V\setminus\{s,t\},\end{equation}
\begin{equation}\label{eq:3}\sum_{(s,v)\in E}\sigma((s,v))-\sum_{(v,s)\in E}\sigma((v,s))=1,\quad \mbox{ and}\end{equation}
\begin{equation}\label{eq:4}\sum_{(t,v)\in E}\sigma((t,v))-\sum_{(v,t)\in E}\sigma((v,t))=-1.\end{equation}
All the above sums can be computed in polynomial time, because the terms involved have size polynomial in $|(G,s,t,k)|$. In particular, each term of the form $\sigma(e)\cdot\mu(e)$ has size at most $|\sigma|\cdot |E|=O(|(G,s,t,k)|^3)$ (recall that, by assumption, $|\sigma|=O(|(G,s,t,k)|^2)$).

Now we prove the correctness of this verifier. Suppose that there is a walk in $G$ from $s$ to $t$ whose total weight is $k$, and hence the instance $(G,s,t,k)$ is positive. Let $w$ be the shortest such walk in terms of hops. We call an edge of $G$ whose weight is $(0,0,\cdots,0)$ a \emph{null edge}. Let $k=(k_1,k_2,\cdots,k_m)$, and let $\tau=k_1+k_2+\cdots +k_m$. Then, the number of times that $w$ traverses a non-null edge is at most $\tau$. Moreover, since $w$ is minimal, it cannot traverse more than $|V|$ consecutive null edges, or else some of them could be deleted without changing $w$'s total weight, nor its starting and ending vertices. It follows that the length of $w$ is at most $(\tau+1)|V|$. For each $e\in E$, if we let $\sigma(e)$ be the number of times that $w$ traverses edge $e$, we have $\sigma(e)\leqslant(\tau+1)|V|$. Note that, if $k^*$ is the maximum of the $k_i$'s, we have $\tau/m\leqslant k^*$, and therefore
$$\log_2\tau=\log_2\frac\tau m+\log_2 m\leqslant\log_2 k^*+\log_2 m\leqslant\sum_{i=1}^m\log_2 k_i+\log_2 m\leqslant |k|+\log_2 m.$$
Hence, recalling that $m$ is a constant, the size of $\sigma$'s representation is at most
$$|E|\cdot(\left\lfloor\log_2(\tau+1)|V|\right\rfloor+1)=O\left(|E|\cdot |k|+|E|\cdot\log|V|\right)=O(|(G,s,t,k)|^2).$$
Moreover, by construction, $\sigma$ passes all the tests, and therefore it is accepted by the verifier.

Conversely, suppose that some function $\sigma\colon E\to\mathbb N$ passes all the tests, causing the verifier to accept. Let us construct a directed multigraph $G'$ as follows. The vertices of $G'$ are the vertices of $G$ that are incident to an edge on which $\sigma$ is greater than zero. Also, for every $(v,u)\in E$, $G'$ has exactly $\sigma((v,u))$ copies of the directed edge $(v,u)$. Because $\sigma$ passes the above tests, $G'$ is weakly connected and satisfies equations~\eqref{eq:1} to~\eqref{eq:4}, implying that it is Eulerian or semi-Eulerian. Therefore $G$ has an Eulerian walk $w'$ from $s$ to $t$. Such an Eulerian walk directly translates into a walk $w$ on $G$ from $s$ to $t$ that traverses each edge $e\in E$ exactly $\sigma(e)$ times. But by assumption $\sum_{e\in E}\sigma(e)\cdot\mu(e)=k$, hence the total weight of $w$ is exactly $k$, which implies that the instance $(G,s,t,k)$ is indeed positive.
\qed
\end{proof}

We are now able to prove the two main theorems of this section.

\begin{theorem}\label{t1}
The decision version of \LEM, restricted to levels in which the Builder skills are polynomially many, belongs to \NP.
\end{theorem}
\begin{proof}
Recall that the total number of lemmings is polynomial in the number of \emph{terrain} cells. It follows that permanent skills, i.e., Climber and Floater skills, can be assigned only polynomially many times, and therefore involve a polynomial number of moves. The same holds, for obvious reasons, also for Bomber skills.

Observe that \emph{terrain} may change only at polynomially many time units. Indeed, \emph{terrain} can be altered only as a consequence of a player's action. Moreover, the only way to create new solid cells is via a Builder, which affects at most 72 cells. Hence the number of cells that can be restored is bounded by a polynomial. As a consequence, also the cells that can be destroyed is bounded by a polynomial.

Digger skills can be assigned only to lemmings that can effectively dig some solid cells, and hence they involve at most a polynomial number of moves.

Blocker skills can also be assigned polynomially many times, because a Blocker can be interrupted only by turning it into a Bomber, or by removing the terrain on which its pin lies.

Let us consider a feasible solution that saves a certain number of lemmings in a given level in the minimum number of time units. We will produce a certificate for such a solution that is verifiable in polynomial time. Observe that, by the above reasoning, there are at most polynomially many time instants at which \emph{terrain} changes, a lemming is eliminated, or a skill other than Basher or Miner is assigned. Part of the certificate for the given solution will be the sequence of all such configurations (plus an initial and a final configuration), each with a timestamp. Note that a complete description of a configuration can be given in polynomial space, because it only includes the state of each \emph{terrain} cell, the state and position of each lemming, the number of available skills, and possibly a player's action.

It remains to show that all the timespans can be encoded in polynomial space. Consider the polynomially many maximal timespans during which \emph{terrain} does not change, no lemming is eliminated, and only Basher and Miner skills are assigned. Let $[t, t']$ be one such maximal timespan (possibly with $t=t'$). Because there exist at most an exponential amount of combined configurations of all lemmings, the configuration at time $t'$ is reached also at some time $t''\leqslant t'$ such that $t''-t$ is bounded by an exponential. Since the solution we are considering is minimal by assumption, $t'=t''$. Therefore the timestamps of all the player's actions are bounded by an exponential in the size of the level (even if the level has no time limit), and hence all of them can be encoded in polynomial space.

Observe that the polynomially many transitions between consecutive time units in which either \emph{terrain} changes, or some lemming is eliminated, or moves other than Basher and Miner are made can be simulated in polynomial time, because each transition involves a constant number of operations and tests for each lemming. The remaining time intervals are polynomially many and may be exponentially long. Let $[t, t']$ be one such interval, and let the certificate contain the descriptions of the configurations at times $t$ and $t'$. We need to show how to certify that $t'$ can be reached from $t$ in the desired amount of time units and using the desired amount of Basher and Miner skills. Between $t$ and $t'$, lemmings do not interfere with each other (recall that \emph{terrain} does not change), and hence each of them follows a path that can be verified independently of the others.

Let us fix a lemming, and let us construct a weighted directed multigraph $G$ on the set of its possible positions and states in the landscape. $G$ has an edge from $v$ to $u$ if the lemming can go from $v$ to $u$ in one time unit (recall that $v$ and $u$ encode positions and states of the lemming). The edge $(v,u)$ has weight $(1,b,m)$, where $b$ and $m$ can take the value $0$ or $1$, depending if the player has to use a Basher or a Miner skill, respectively, to make the lemming go from $v$ to $u$. (Recall that Basher and Miner skills can occasionally be assigned to delay a Walker for a couple of time units, or to reverse its direction, without removing \emph{terrain} cells.) Clearly, $G$ can be constructed in polynomial time, because a lemming only has polynomially many possible positions and a constant number of possible states. The lemming also has a starting position and state, taken from the snapshot at time $t$, and a final position and state, taken from the snapshot at time $t'$ (since lemming are indistinguishable, the certificate must supply this information, as well). The number of time units between $t$ and $t'$ is simply $t'-t$, and the number of Basher and Miner skills that have to be assigned to each lemming between $t$ and $t'$ is supplied by the certificate (this information can be verified by simply adding together all the Basher and Miner skills that have to be assigned to individual lemmings, and checking if the sum is equal to the difference of available Basher and Miner skills in the configurations at times $t'$ and $t$).

Hence, what the verifier should assess is if there exists a walk in $G$ with assigned starting and ending vertices, and assigned total weight, express as the number of time units spent, and the number of Basher and Miner skills assigned. This is an instance of \WGW, for which there is a polynomial-time-verifiable certificate, due to Lemma~\ref{lwgw}.

Summarizing, the certificate includes polynomially many ``snapshots'' of the level taken at different times. Moreover, for each lemming, the certificate contains information oh how to make it move from each snapshot to the next in the correct amount of time and using the correct amount of skills. Since each piece of information has polynomial size and the lemmings are polynomially many, the certificate can be verified in polynomial time.
\qed
\end{proof}

\begin{theorem}\label{t1b}
The decision version of \LEM, restricted to levels in which the Basher, Miner, and Digger skills are polynomially many, belongs to \NP.
\end{theorem}
\begin{proof}
The proof follows the lines of the proof of Theorem~\ref{t1}. Climber, Floater, and Bomber skills can be assigned only once per lemming, hence they involve at most a polynomial number of actions. It follows that \emph{terrain} cells can turn from solid to empty at most polynomially many times (note that each assignment of a Bomber, Basher, Miner, or Digger skill can remove only polynomially many cells), and therefore also the Blocker skills that can be assigned are polynomially many (recall that a Blocker can be interrupted only by eliminating it or removing the solid cell on which it stands).

Once again, we are left with polynomially many timespans during which \emph{terrain} does not change, no lemming is eliminated, and only Builder skills can be assigned to ``ineffectively'' lay one or more bricks of a stairway and possibly make a lemming turn around. A certificate for these timespans exists due to Lemma~\ref{lwgw}, while the snapshots of the polynomially many other configurations can be given explicitly in the certificate.
\qed
\end{proof}

\section{\PSPACE-Complete Instances}\label{s4}
Next we show that there are classes of levels of \LEM that are \PSPACE-hard. Due to Theorems~\ref{t1} and~\ref{t1b}, it comes as no surprise that such levels have exponentially many (or unlimited) available Builder and Basher skills. As a bonus, there is only one lemming per level.

For the \PSPACE-hardness reduction, we apply a framework introduced in~\cite[Metatheorem~4.c]{viglietta}, which has also appeared in~\cite{nintendo}. The framework is based on a reduction from \QBF involving a player-controlled \emph{avatar}, a \emph{starting location}, an \emph{exit location}, several \emph{paths}, \emph{pressure plates} and \emph{doors}.

Doors are areas that can be traversed by the avatar if and only if they are \emph{open}. For each door, there are exactly two pressure plates located somewhere in the level, which open and close the door, respectively. Pressure plates are activated whenever the avatar traverses them, but we observe that the proof of~\cite[Metatheorem~4.c]{viglietta} goes through even if the pressure plates that \emph{open} doors are implemented as \emph{buttons}, i.e., the avatar is not forced to activate them upon traversal, but may or may not do it, at the player's will. Indeed, opening a door has the only effect of expanding the set of locations that can be reached by the avatar, so it is never ``wrong'' to do it whenever possible.

Hence we may implement a pressure plate ``indirectly'', as a path that starts from the location that should contain it, reaches the corresponding door gadget from the proper direction, and then leads back to where it started. In other words, we can incorporate pressure plates within their respective door gadgets. The new specification of door that we obtain is as follows: a door is a device that can be either open or closed, and is crossed by three separate paths:
\begin{itemize}
\item a \emph{traverse path}, which can be traversed by the avatar if and only if the door is in the open state;
\item a \emph{close path} which, when traversed by the avatar, causes the door to close;
\item an \emph{open path} which, when traversed by the avatar, allows the player to open the door, but does not force them to.
\end{itemize}
To better distinguish these doors from the pressure-plate-operated doors of~\cite{viglietta}, we call them \emph{stand-alone doors}.

Different stand-alone doors, along with the starting and exit locations, are connected by paths. These should be traversable by the avatar arbitrarily many times in at least one direction, and the avatar should not be able to move from one path to another unless they merge at a point (note that to solve the levels constructed in the proof of~\cite[Metatheorem~4.c]{viglietta}, the avatar never has to traverse the same path in two opposite directions). Also, when a path forks, the player should be able to choose which way the avatar goes, and the choice may change each time the avatar gets there. Paths should be able to ``cross'' each other without allowing the avatar to change direction when it reaches a crossroads. Note that such a \emph{crossover gadget} is not needed with pressure-plate-operated doors, but it is with stand-alone doors. Finally, the player may even be allowed to ``break'' a path, making it impassable by the avatar, as long as breaking it does not cause leakage to other paths.

With the above definitions, we can re-formulate~\cite[Metatheorem~4.c]{viglietta} as follows.

\begin{metatheorem}\label{meta}
If a game features stand-alone doors, paths with crossovers, and the player must guide an avatar from a starting location to an exit location, then the game is \PSPACE-hard.\hfill\qed
\end{metatheorem}

Now we are able to prove the main theorem of this section.

\begin{theorem}\label{t2}
\LEM is \PSPACE-complete, even restricted to levels with only one lemming, and only Builder and Basher skills.
\end{theorem}
\begin{proof}
The membership in \NPSPACE of \LEM is obvious, because each game configuration can be stored in polynomial space, and the configuration graph can be efficiently navigated (cf.~the proof of Theorem~\ref{t1}). The membership in \PSPACE thus follows from Savitch's Theorem (see~\cite{papadimitriou}).

For the \PSPACE-hardness part we implement the previously described framework, summarized by Metatheorem~\ref{meta}. In our implementation, the only lemming in the level will be the avatar, which will be controlled by the player via the assignment of Builder and Basher skills at very specific locations. We build the level in such a way that every $4\times 4$ cell is made of steel by default, unless it contains the trigger area of a deadly zone (recall from Section~\ref{s2} that deadly zones must be permeable).

Since the starting and exit locations are trivial to implement in Lemmings, we only have to show how to construct paths with forks and intersections, and how to construct stand-alone doors.

\begin{figure}[h]
\centering
\subfigure[]{\label{f2a}\includegraphics[scale=0.8]{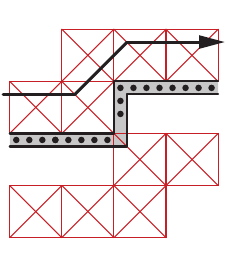}}\qquad\ 
\subfigure[]{\label{f2b}\includegraphics[scale=0.8]{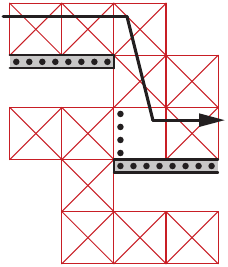}}\qquad\ 
\subfigure[]{\label{f2c}\includegraphics[scale=0.8]{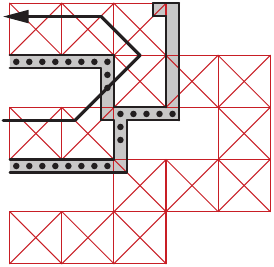}}\\
\subfigure[]{\label{f2d}\includegraphics[scale=0.8]{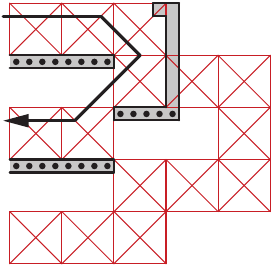}}\qquad\qquad
\subfigure[]{\label{f2e}\includegraphics[scale=0.8]{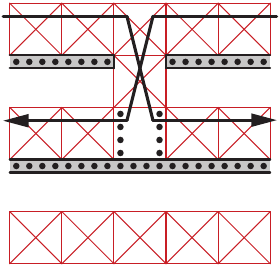}}
\caption{Paths. The lemming's pin must never touch the deadly zones.}
\label{f2}
\end{figure}

Figure~\ref{f2} shows how paths are implemented. White space denotes empty \emph{terrain} cells, gray space denotes solid cells, and black dots mark the positions occupied by the lemming's pin as the paths are traversed following the arrows. Each large crossed square represents a $4\times 4$ block containing the trigger area of a deadly zone. Collectively, deadly zones prevent the lemming from deviating from its path. Note that the one in Figure~\ref{f2b} is also a one-way path, because it cannot be traversed from right to left. To complete the intersection gadget in Figure~\ref{f2e}, we attach a one-way path to each side of it: one copy is attached to the lower-right part, and a symmetric copy to the lower-left part. This is to prevent the lemming from accidentally reversing direction after the intersection (perhaps when reaching a different gadget and acting ``inappropriately''), and then reenter the intersection gadget and take the wrong path.

Observe that, conceivably, when the lemming lands on the bottom platform in the gadgets in Figures~\ref{f2b} and~\ref{f2e}, it could start building a stairway there, because there is no deadly zone immediately above that platform. However, a Builder cannot be interrupted in that location by turning it into a Basher, because there is steel all around the deadly zones (cf.~Figure~\ref{f1c}). This is true unless the Builder has laid at least two steps of the stairway. But if it does so, it inevitably enters a deadly zone, no matter where it started building. Also note that some of the paths may be broken if the lemming becomes a Basher at the right time and removes some solid cells. However, this action has the only possible effect of rendering some paths unusable, and therefore it does not invalidate the construction.

\begin{figure}
\centering
\subfigure[Initial configuration]{\label{ffa}\includegraphics[scale=0.65]{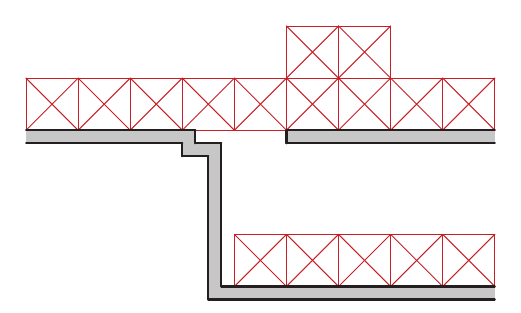}}\quad\ 
\subfigure[Building a bridge to the top path]{\label{ffb}\includegraphics[scale=0.65]{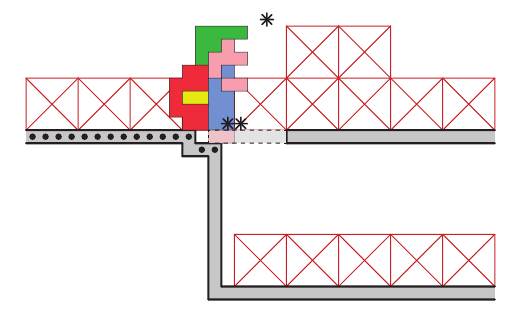}}\\
\vspace{0.75cm}
\subfigure[Bashing to avoid building further]{\label{ffc}\includegraphics[scale=0.65]{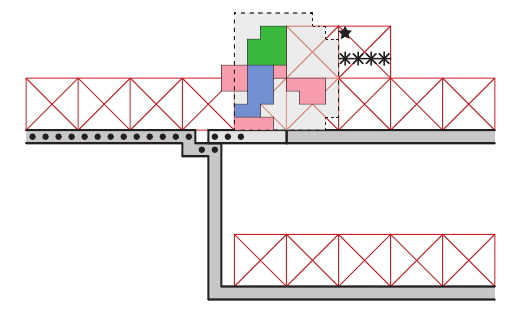}}
\quad\ 
\subfigure[Proceeding to the top path]{\label{ffd}\includegraphics[scale=0.65]{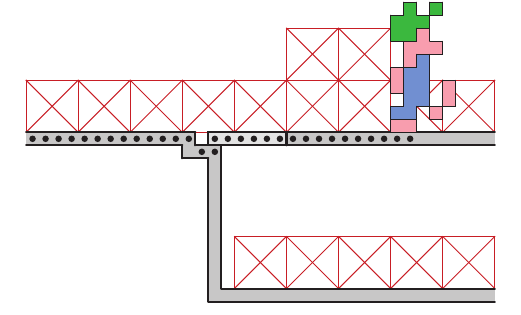}}\\
\vspace{0.75cm}
\subfigure[Removing the bridge]{\label{ffe}\includegraphics[scale=0.65]{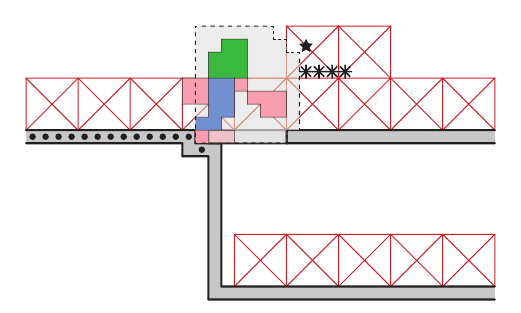}}\quad\ 
\subfigure[Proceeding to the bottom path]{\label{fff}\includegraphics[scale=0.65]{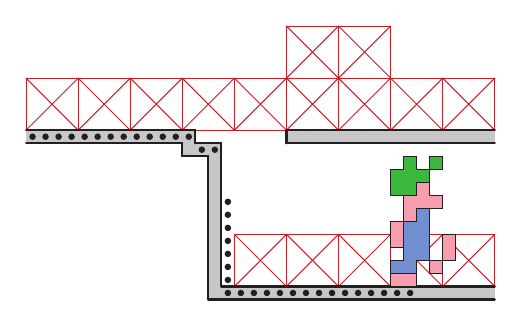}}
\caption{Fork gadget}
\label{f3a}
\end{figure}

It is easy to see that any directed graph can be embedded in the plane by suitably arranging copies of the five gadgets in Figure~\ref{f2} and their mirror images, provided that paths can also be forked and merged together. Two paths can be merged via a gadget similar to the one in Figure~\ref{f2b}: it is sufficient to extend the bottom platform toward the left, and lower the deadly zones that are in the way. Forks are implemented via the gadget in Figure~\ref{ffa}. The lemming enters from the left, and then the player may direct it to any of the two exits on the right. The deadly zones are the only permeable blocks (recall that all non-deadly blocks are made of steel in our construction), and they are positioned in such a way that a Builder can lay a single brick of a stairway (Figure~\ref{ffb}), climb on it, and then immediately become a Basher to stop laying further bricks (Figure~\ref{ffc}) and proceed to the right as a Walker (Figure~\ref{ffd}). Moreover, if a brick is already present when the lemming arrives, it can be removed by assigning a Basher skill right before the lemming climbs on it. This will cause the lemming to excavate precisely the 6-cell brick with one stroke (Figure~\ref{ffe}) and then fall down and proceed to the right as a Walker (Figure~\ref{fff}). Any other way of assigning skills is either ineffective, or deadly, or prevents the lower area from being reached, which never helps the lemming reach its final goal.

The stand-alone door implementation is depicted in Figure~\ref{fda}, where the door is considered open if and only if the central dashed rectangle is made of empty cells. The picture shows the locations where the traverse, open, and close paths are connected to the gadget, and the possible trails of the avatar's pin. A one-way path (not shown in the picture) is also attached after each of the three paths, to prevent the lemming from entering the gadget from the wrong direction.

\begin{figure}
\centering
\subfigure[Paths and possible trails]{\label{fda}\includegraphics[scale=0.7]{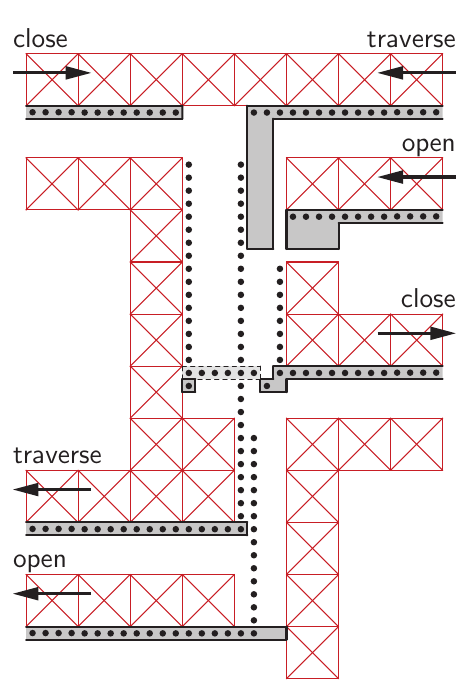}}\qquad
\subfigure[Bashing after closing the door]{\label{fdb}\includegraphics[scale=0.7]{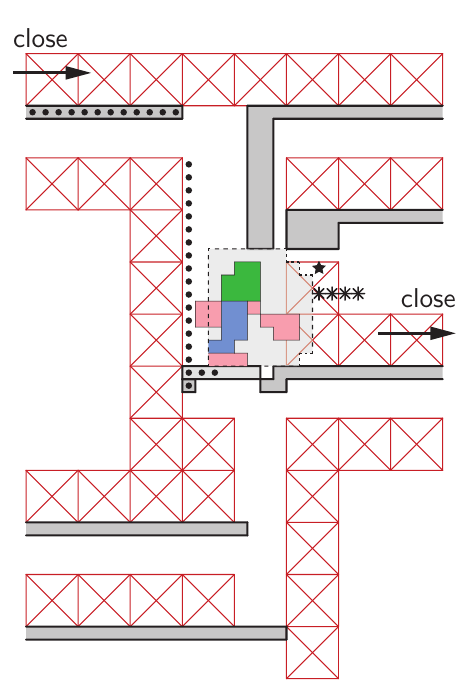}}\\
\vspace{0.5cm}
\subfigure[Opening the door]{\label{fdc}\includegraphics[scale=0.7]{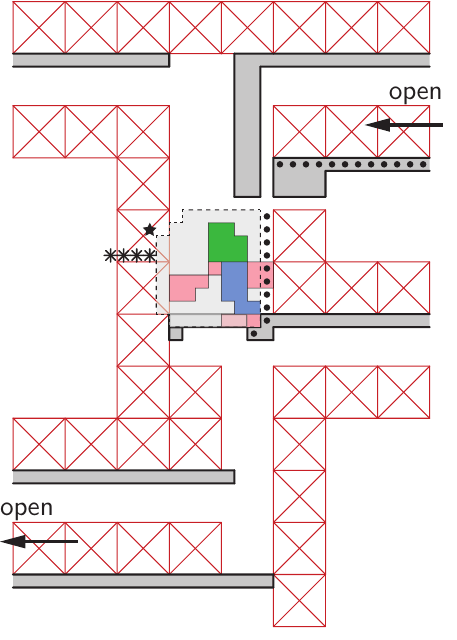}}
\qquad
\subfigure[Trying to traverse a closed door]{\label{fdd}\includegraphics[scale=0.7]{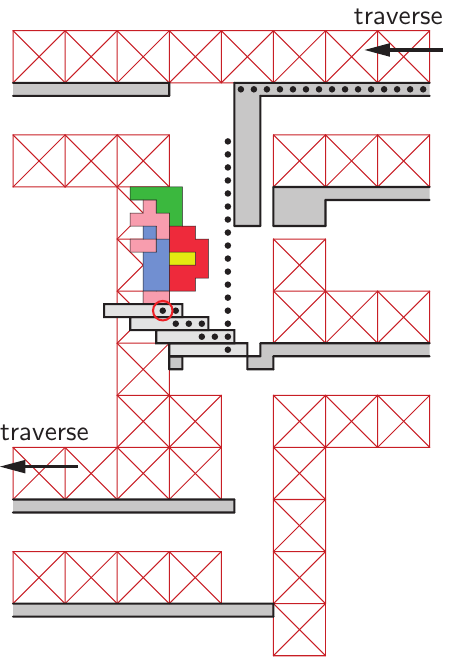}}
\caption{Door gadget}
\label{f3b}
\end{figure}

If the lemming is traversing the ``close'' path and the door is open, then it must lay a stairway brick (to avoid falling into the deadly zones underneath), thus closing the door, and then it must stop immediately, using a Basher skill, in order to proceed to the right without ``hitting the ceiling'' and turning around (Figure~\ref{fdb}). If it does anything different, it is bound to enter some deadly zone and perish. In particular, it is straightforward to see that if it tries to lay an additional stairway brick when the door is already closed, then it cannot become a Basher because the permeable block is too low, so it eventually hits the ceiling, turns left and dies in the deadly zones no matter what it does.

If the lemming is traversing the ``open'' path, then it is allowed to open the door by using a Basher skill, right before falling down, thanks to the permeable blocks on the left (Figure~\ref{fdc}). If it ever becomes a Builder, then it is bound to die in a few time units, as it can be easily verified.

When the lemming actually attempts to cross the door through the ``traverse'' path, it survives if and only if the door is open. Again, becoming a Builder at any time would kill it, no matter what it does next (Figure~\ref{fdd}).

This completes the construction. Each of these levels is either unsolvable with any amount of Builder and Basher skills, or solvable with exponentially many of them. Deciding whether it is one or the other is \PSPACE-complete.
\qed
\end{proof}

\section{Inapproximability}\label{s5}
Here we consider the restriction of \LEM to instances with only one type of skill. Therefore we have eight possible variations, one for each skill type: \CLI is the variation with only Climber skills, \FLO the variation with only Floater skills, and so on. By Theorems~\ref{t1} and~\ref{t1b}, the decision versions of all these problems are solvable in {\NP}\@. In this section we show that, as optimization problems, all of them are \APX-hard. Note that, due to~\cite[Theorem~2]{cormode}, \CLI and \FLO restricted to level with no deadly zones are solvable in {\P}\@. Hence, what makes \CLI and \FLO levels hard is precisely the presence of deadly zones.

We first introduce the \emph{2-choice gadget}. This will be used as a building block in our reductions. A 2-choice gadget is a device that can produce a single lemming, letting it out from one of two possible exits, named A and B. The player can decide whether the lemming will come out from A or B (or perhaps not come out at all), by doing the appropriate actions at the right times. We first prove that there are 2-choice gadgets for all the variants of \LEM we are considering.

\begin{lemma}\label{l2cho}
There are 2-choice gadgets for \CLI, \FLO, \BOM, \BLO, \BUI, \BAS, \MIN, and \DIG.
\end{lemma}
\begin{proof}
A 2-choice gadget for \CLI and \BUI is shown in Figure~\ref{f2c1}. The lemming enters the level from an entrance not shown in the picture, and keeps walking back and forth until the player decides to assign it a skill to let it out. Depending on when the skill is assigned, the lemming will exit either from A or from B.

\begin{figure}[h]
\centering
\includegraphics[scale=0.8]{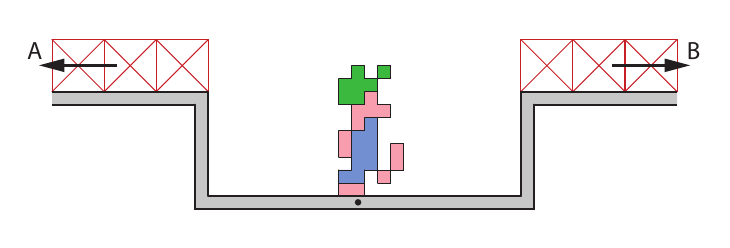}
\caption{2-choice gadget for \CLI and \BUI}
\label{f2c1}
\end{figure}

For all the \emph{terrain}-removing skills, the 2-choice gadget is the same, and it is illustrated in Figure~\ref{f2c2}, where the cells marked by diagonal lines are permeable. In \BAS, \MIN, and \DIG, the player can make the lemming dig through the left or the right side of the gadget, making it exit from A or B. In \BOM there are two lemmings instead of one (the second lemming is not shown in the picture), so the player can make one of them explode at the right time to let the other lemming escape from one of the two sides.

\begin{figure}[h]
\centering
\includegraphics[scale=0.8]{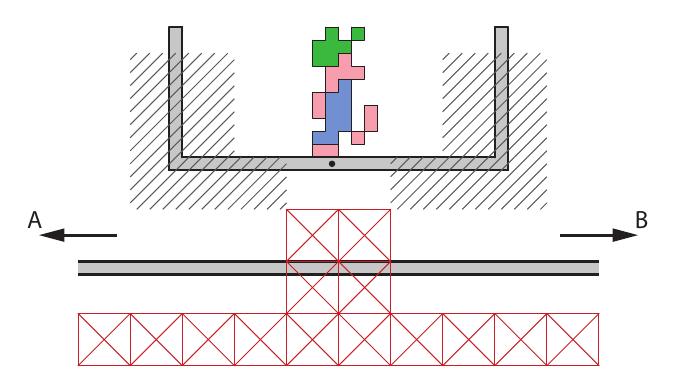}
\caption{2-choice gadget for \BOM, \BAS, \MIN, and \DIG}
\label{f2c2}
\end{figure}

The 2-choice gadget for \BLO is given in Figure~\ref{f2c3}. Depending on when the player assigns a Blocker skill to the front lemming, the rear lemming may exit from A or B (or remain stuck forever in the top-left part of the gadget). If the player does anything else, at least one lemming dies in the deadly zone on the right, while the other lemming possibly exits from B.

\begin{figure}[h]
\centering
\includegraphics[scale=0.8]{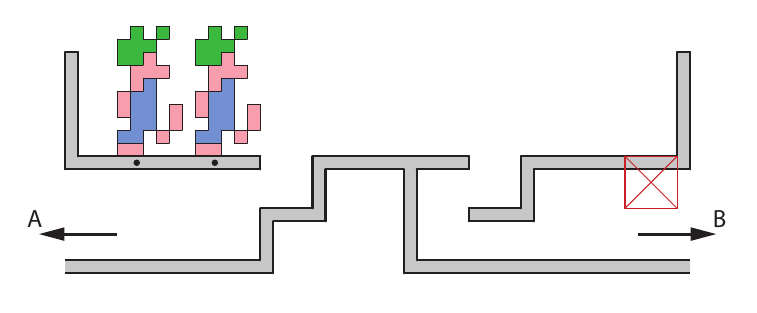}
\caption{2-choice gadget for \BLO}
\label{f2c3}
\end{figure}

Finally, Figure~\ref{f2c4} sketches a 2-choice gadget for \FLO. If the player lets the lemmings go without doing anything, the lemming on the right will reach the deadly zone first, and it will die. The deadly zone will remain harmless for a few time units, and the left lemming will be able to traverse it safely, exiting from B. On the other hand, if a Floater skill is assigned to the right lemming, its fall is slowed down, and the left lemming reaches the deadly zone first. In this case, the right lemming survives and exits from A.

\begin{figure}[h]
\centering
\includegraphics[scale=0.8]{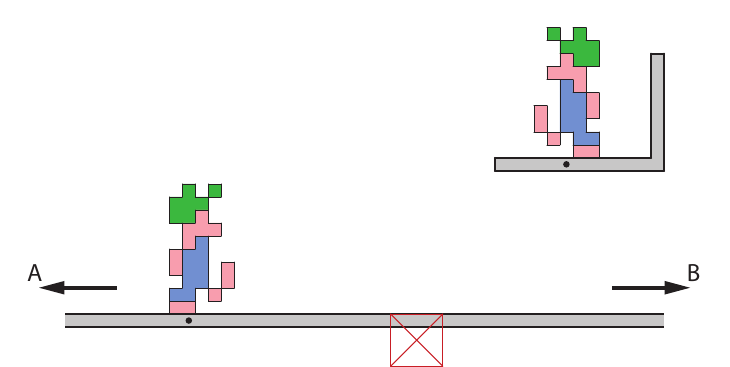}
\caption{2-choice gadget for \FLO}
\label{f2c4}
\end{figure}
\qed
\end{proof}

Next we show that there are \emph{paths} for all our variations of \LEM. The purpose of a path is to connect two different parts of a level. The same path can be traversed several times by lemmings going in one direction, and it should prevent leakage with other paths, no matter what the player does. Paths should also be able to cross each other, again with no leakage. We allow the player to ``destroy'' a path, making it unusable, as long as destructing it does not cause any leakage with other paths. In contrast with Section~\ref{s3}, here paths are not required to fork.

\begin{lemma}\label{lpath}
There are paths for \CLI, \FLO, \BOM, \BLO, \BUI, \BAS, \MIN, and \DIG.
\end{lemma}
\begin{proof}
The first four gadgets in Figure~\ref{f2} work also in this case, with minor modifications. In fact, the deadly zones surrounding the paths prevent the lemmings from going anywhere else, no matter what skills the player assigns them (we have to assume the delay parameter of each deadly zone to be $0$, since several lemmings may be traversing a path at the same time). Hence, we just have to check that the paths can indeed be safely traversed by lemmings. This is true for Walkers, as Figure~\ref{f2} shows. However, the third and fourth gadget do not work for Climbers, because they climb the wall on the right instead of turning around, and they perish in the deadly zones. Notice, though, that these deadly zones are not required in \CLI, and so they can be removed. Now, a lemming climbing the right walls will hit the ceiling on top (notice the protruding solid cell), turn around, fall down, and proceed as intended.

The intersection gadget in Figure~\ref{f2e} requires a more careful analysis. In \CLI and \FLO, it clearly works as intended. In \BOM, the player could destroy it by making a lemming explode in midair, but that would only make the gadget unusable, if anything. In \MIN and \DIG, either a lemming would not be able to dig in a certain spot due to steel, or it would dig some ground and make part of the gadget unusable. The same goes for \BAS, except that a lemming could either be unable to bash in a certain spot, or bash ``ineffectively'', i.e., without altering \emph{terrain} or changing direction. In \BUI, it is easy to see that a lemming could build part of a stairway when it lands on the lower platform. Then it would die in a deadly zone, but this would effectively raise the lower platform by one cell. A second lemming would be able to raise it further, etc. However, this can clearly happen at most four times, and then any lemming reaching the gadget would either be forced to traverse it as intended, or it would pass through a deadly zone and perish.

The only case left to consider is \BLO, and unfortunately the intersection gadget in Figure~\ref{f2e} does not work here, because turning a lemming into a Blocker in the lower platform would make subsequent lemmings turn around and take the wrong path. Hence we replace that gadget with the one in Figure~\ref{f2c4}.

\begin{figure}[h]
\centering
\includegraphics[scale=0.8]{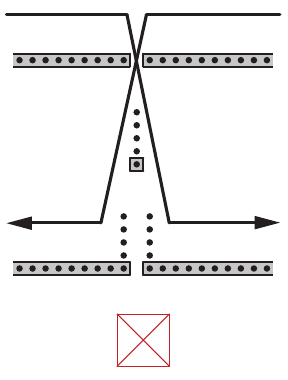}
\caption{Intersection gadget for \BLO}
\label{f2c4}
\end{figure}

Now, if the player puts a Blocker on the lower platform, the lemmings that turn around at it are either bound to fall into the deadly zone or become Blockers themselves (and remain Blockers forever). Even if the player puts a Blocker on the central isolated cell, that Blocker is ineffective, because it cannot interact with the lemmings on the upper platform (as they are out of range), or with the lemmings that fall on the same square.
\qed
\end{proof}

With the gadgets constructed so far, we are able to give a unified \APX-hardness reduction that works for all eight variations of \LEM.

\begin{theorem}\label{t3}
\LEM is \APX-hard, even restricted to levels in which only one type of skill is available (whatever the skill is).
\end{theorem}
\begin{proof}
For each of the eight possible restrictions of \LEM to levels with only one skill type, we give an L-reduction from the \APX-complete problem \MTSAT (refer to~\cite{ausiello}), in which the number of satisfied clauses of a 3-CNF Boolean formula has to be maximized. We describe a unified reduction framework that only uses the 2-choice gadgets and the paths constructed in Lemmas~\ref{l2cho} and~\ref{lpath}. The framework is the same, regardless of what skill is available to the player.

We need \emph{variable gadgets} and \emph{clause gadgets}, wired together with paths. A gadget for variable $x$ is sketched in Figure~\ref{f4a}.
\begin{figure}[h]
\centering
\includegraphics[scale=0.4]{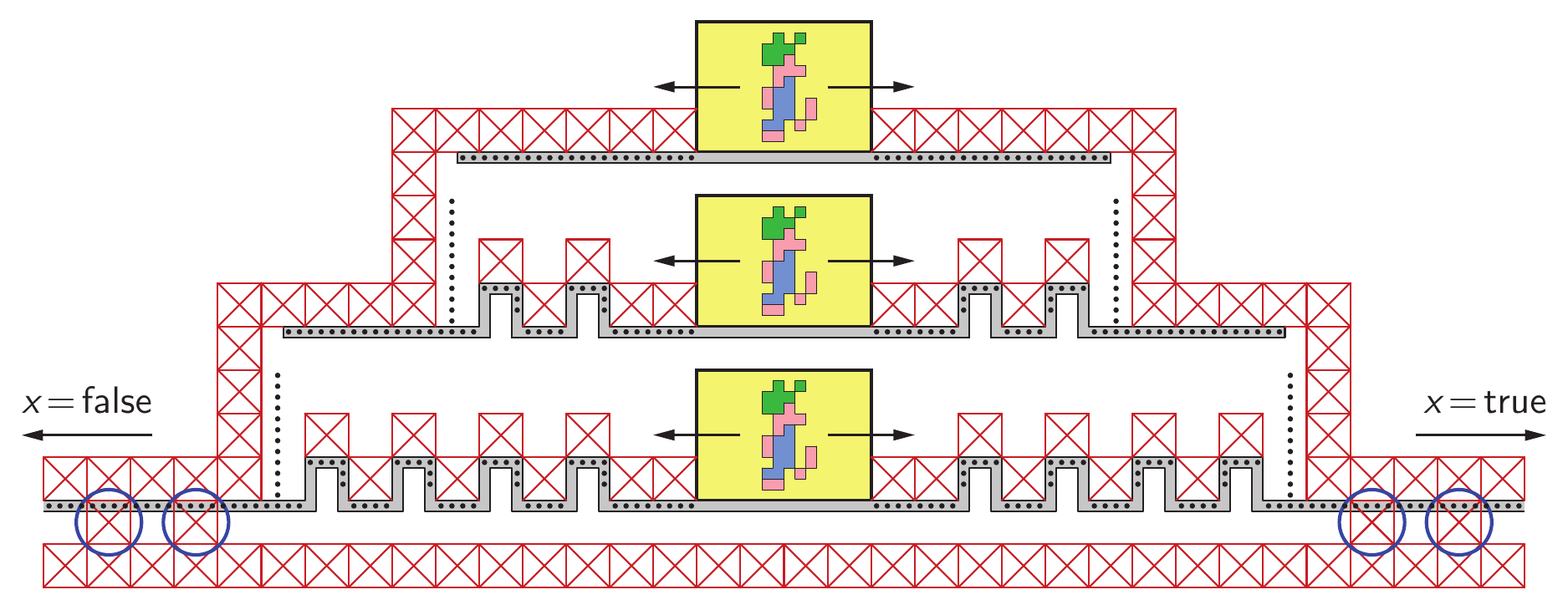}
\caption{\APX-hardness variable gadget}
\label{f4a}
\end{figure}
It is made of $2k+1$ \emph{layers}, where $k$ is the number of occurrences of variable $x$ in the formula, and each layer contains a 2-choice gadget. (In the figure, $k=1$, and each 2-choice gadget is represented by a rectangle containing a lemming.) Each lemming can exit from the left side or from the right side of its respective 2-choice gadget. All lemmings exiting from the same side eventually reach a common path, on which a row of $k+1$ traps is found, each with a delay parameter of a few time units (these traps are circled in the figure). The exact value of the delay parameter is irrelevant, as long as it is a large-enough constant.

Because each trap kills at least one lemming upon traversal, no two lemmings can exit the variable gadget from different sides and survive, since that would kill at least $2k+2$ lemmings. However, if the player makes all the 2-choice gadgets produce a lemming from the same side, then exactly $k$ lemmings safely exit the variable. This is because layers have different amounts of small \emph{bumps} whose purpose is to ``synchronize'' lemmings. Since lemmings falling from higher layers take a longer route, those on the lower layers are slowed down by the bumps, in such a way that all lemmings eventually reach the row of $k+1$ traps approximately at the same time, and each trap kills exactly one lemming.

Note that the player can conceivably destroy some of the bumps or parts of the layers, or put Blockers somewhere, etc. However, this deviant behavior only gets more lemmings killed, and it can in no way allow lemmings to safely exit the variable gadget from two different sides.

So, the truth value of $x$ will be encoded by the side from which $k$ (or fewer) lemmings exit the corresponding gadget. After coming out of the true (respectively, false) side of the variable gadget, the group of lemmings reaches, one by one, all the clause gadgets containing a positive (respectively, negative) occurrence of $x$. The group eventually falls into a pool of water, so that any remaining lemming is killed.

In Figure~\ref{f4b}, a gadget for clause $(\ell_1 \vee \ell_2 \vee \ell_3)$ is sketched, where once again each rectangle containing a lemming represents a 2-choice gadget.
\begin{figure}[h]
\centering
\includegraphics[scale=0.4]{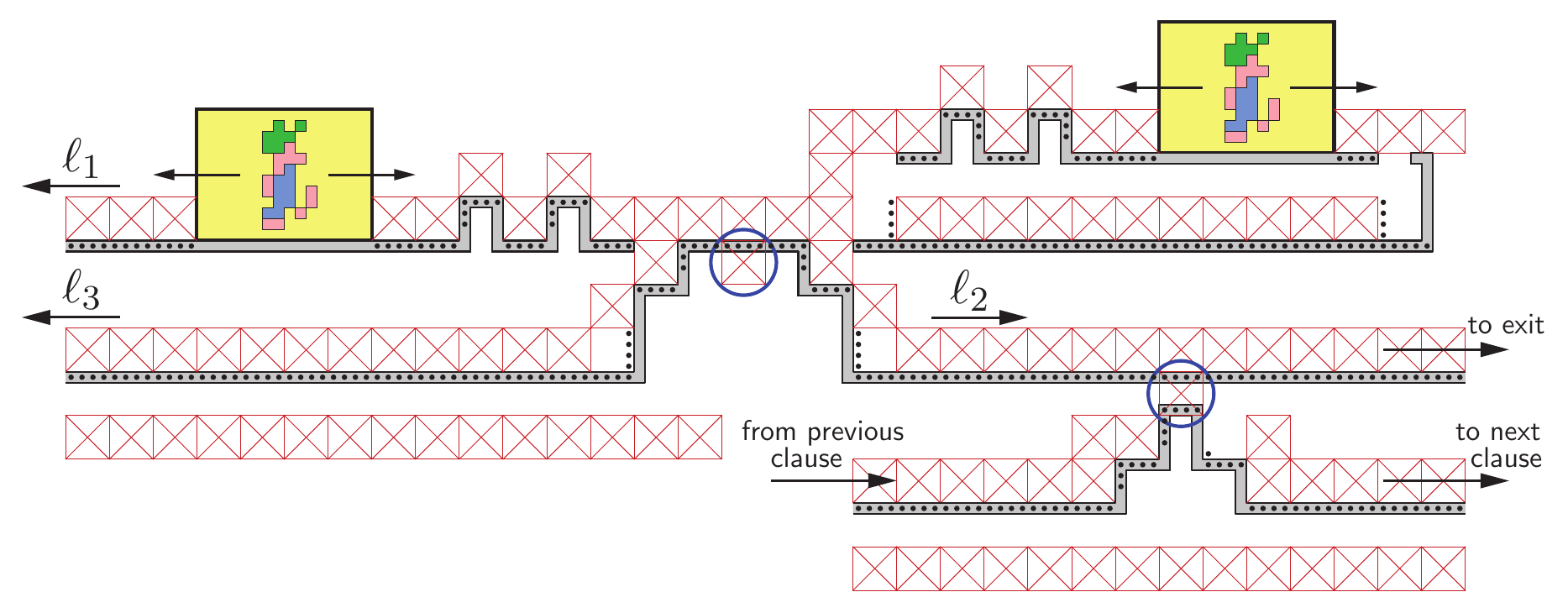}
\caption{\APX-hardness clause gadget}
\label{f4b}
\end{figure}
The clause gadget has three exits, one for each literal in the clause. If the player makes the left 2-choice gadget produce a left-facing lemming, it exits the gadget from the path corresponding to $\ell_1$. Meanwhile, no matter which direction the lemming exiting the other 2-choice gadget is facing, it is bound to die in the central circled trap (or become a Blocker or a Bomber, which effectively eliminates it, or dig into some \emph{terrain} cells, thus falling into some deadly zone, etc.). On the other hand, if the lemming produced by the leftmost 2-choice gadget is facing right, it will fall into the central trap (again, unless it is eliminated before). If the second lemming is facing left, it will reach the trap first, making it temporarily harmless, and allowing the first lemming to reach the $\ell_2$ exit. Otherwise, if the second lemming is facing right, it will take a longer route, and it will reach the trap shortly after the first lemming, thus surviving and reaching the $\ell_3$ exit. There are sequences of bumps along these paths, which delay the lemmings in such a way that they reach the central trap at the proper times. In other words, this is a \emph{3-choice gadget}, which allows the player to choose from which of three possible paths a single lemming will exit.

After each of the three exits, another trap is encountered (only shown for literal $\ell_2$ in Figure~\ref{f4b}), and then a path safely leads to a level exit. The same trap is also traversed by a path underneath, coming from the variable gadget corresponding to the same literal. This is done in such a way that the clause lemming can be saved by the group of variable lemmings if and only if the literal is true according to the chosen assignment.

Observe that there is no way for the player to make any variable lemming reach the exit, unless Builders are available. In \BUI, the player could conceivably make several variable lemmings quickly start building stairways on top of each other while the trap is still harmless, and some of them may be able to reach the upper path that leads to the exit. However, this can be prevented by ``switching'' the two paths, i.e., placing the variable path above the literal path, in such a way that no variable lemming could reach the literal path, no matter how Builder skills are assigned.

In order to guarantee synchronization and make sure that the trap is reached by the head of the group of variable lemmings just a couple of time units before the upcoming clause lemming, the variable and literal paths may have to be modified (this is not shown in the figure). It is easy to see that, if one of the two paths is too short, it can be stretched by scaling up the whole construction and making the path ``squiggle'', which effectively multiplies the ratio of the two path lengths by a constant factor. After that, for finer synchronization adjustments, a series of bumps can be added to either path.

Note that in Lemmings at most one action per time unit can be performed by the player. This means that the player cannot interact with two 2-choice gadgets exactly at the same time. However, note that a 2-choice gadget can be made arbitrarily wider without losing its relevant properties. In particular, if the width of the $i$-th 2-choice gadget is proportional to $i$, the player can activate 2-choice gadgets sequentially, each at a different time unit. Moreover, while constructing the paths and adjusting their lengths, we may assume that the player will interact with each 2-choice gadget at one of only two possible specific time units, depending on which side of the gadget the lemming will exit. This makes the construction not depend on the actual choice that the player will make at each 2-choice gadget.

Clearly, only the clause lemmings can possibly be saved in these levels, and each of them may indeed be saved if and only if at least one literal of its clause gadget is true according to the assignment encoded by the corresponding variable lemmings. Therefore, the reduction preserves the optimal value, and any solution that saves $n$ lemmings in one of these levels can be trivially converted into a variable assignment that satisfies exactly $n$ clauses of the corresponding Boolean formula. As a consequence, this is an L-reduction.
\qed
\end{proof}

It follows that the optimal number of saved lemmings is not approximable with a too small relative error, even in these severely restricted sub-games of Lemmings.

\begin{corollary}\label{cor1}
Computing approximate solutions to \LEM with a relative error lower than $1/8$ is \NP-hard, even for levels with only one type of skill (whatever the skill is).
\end{corollary}
\begin{proof}
The proof of Theorem~\ref{t3} describes a mapping from a 3-CNF Boolean formula to a level of \LEM where only a given type of skill is available, such that $n$ clauses of the formula are satisfiable if and only if $n$ lemmings can be saved. It follows that each of the eight sub-games of \LEM has the same inapproximability features of \MTSAT. More formally, we gave an L-reduction with $\beta=\gamma=1$ (refer to~\cite[Definition~8.4]{ausiello}) from \MTSAT, and therefore our claim follows from~\cite[Theorem~6.1]{hastad}, which states that \MTSAT is \NP-hard to approximate with a relative error lower than $1/8$.
\qed
\end{proof}

\section{Conclusions and Open Problems}\label{s6}

In this paper we studied the computational complexity of a generalization of the Lemmings video game. We modeled it in Section~\ref{s2}, as an optimization problem in which the maximum number of lemmings have to be saved. In the model we retained all the mechanics of the original game (including its glitches), while allowing arbitrarily large levels containing arbitrarily many objects.

Previous papers on Lemmings studied some simplified versions of the game, containing only a subset of its elements and simplified mechanics. The game was known to be \NP-hard, even restricted to levels with only one lemming, and it was known to be in \NP if the time limit is polynomial in the size of the level. It was also conjectured to be \PSPACE-complete.

In Section~\ref{s3} we extended the membership in \NP to all levels in which either the Builder skills are polynomially many, or the Bashers, Miners, and Diggers skills combined are polynomially many (regardless of the time limit). This result can be interpreted in the following way: Bashers, Miners, and Diggers can ``undo'' the actions of Builders, in that they destroy solid cells, while Builders make new ones. If only one of these two types of actions can be performed exponentially many times, then the solution of a given level has a short description, i.e., the game belongs to {\NP}\@. On the other hand, if solid cells can be created and destroyed exponentially many times, then a level's solution is not expected to have a short description, i.e., the game is expected to be \PSPACE-complete. In Section~\ref{s4} we proved that this is indeed the case: if exponentially many Builders and Bashers are available, Lemmings is \PSPACE-complete, even restricted  to levels with only one lemming.

Although we settled the most compelling open problem concerning the Lemmings game, other interesting questions remain unanswered. Indeed, we know that Lemmings is \NP-complete even restricted to levels with no deadly zones, while our \PSPACE-completeness reduction abundantly uses them. Therefore, as Table~\ref{tab3} suggests, there is still a gap between \NP-complete and \PSPACE-complete instances. We seek to determine the complexity of the restriction of Lemmings to levels with no deadly zones, but with arbitrary time limit and arbitrarily many available skills. In fact, the complexity of this sub-game is unclear even if just one lemming is present in each level. The absence of deadly zones gives the lemming a much greater freedom of movement, and possibly the ability to reach an easily-determined portion of the level, given enough Builder, Basher, Miner, and Digger skills. It would not be surprising if this sub-game was actually in \P, despite being \NP-complete if only Digger skills are available.

\begin{table}[h]
\begin{center}
{\renewcommand{\arraystretch}{1.25}
\begin{tabular}{c||c|c|c||c|c|c}
Climbers & $\infty$ & $\infty$ & $\infty$ & $\infty$ & $\infty$ & $\infty$ \\
\hline
Floaters & $\infty$ & $\infty$ & $\infty$ & $\infty$ & $\infty$ & $\infty$ \\
\hline
Bombers & $\infty$ & $\infty$ & $\infty$ & $\infty$ & $\infty$ & $\infty$ \\
\hline
Blockers & $\infty$ & $\infty$ & $\infty$ & $\infty$ & $\infty$ & $\infty$ \\
\hline
Builders & poly & $\infty$ & $\infty$ & poly & $\infty$ & $\infty$ \\
\hline
Bashers & $\infty$ & $\infty$ & $\infty$ & $\infty$ & $\infty$ & $\infty$ \\
\hline
Miners & $\infty$ & $\infty$ & $\infty$ & $\infty$ & $\infty$ & $\infty$ \\
\hline
Diggers & $\infty$ & $\infty$ & $\infty$ & $\infty$ & $\infty$ & $\infty$ \\
\hline\hline
Time & $\infty$ & $\infty$ & $\infty$ & $\infty$ & $\infty$ & $\infty$ \\
\hline
Hazards & & & \cmark & & & \cmark \\
\hline
1 lemming & & & & \cmark & \cmark & \cmark \\
\hline\hline
& \NP-complete & \textbf{?} & \PSPACE-complete & \NP-complete & \textbf{?} & \PSPACE-complete
\end{tabular}}
\end{center}
\caption{Some open problems}
\label{tab3}
\end{table}

Finally, in Section~\ref{s5} we proved that the maximum number of lemmings that can be saved in a given level is \NP-hard to approximate with a relative error lower than $1/8$. This holds even if only one skill type is available, whatever the skill is. This result is especially interesting in the case of Climber and Floater skills, because it was previously known that levels with only Climbers and Floaters are solvable in \P, provided that they contain no deadly zones. Hence, we conclude that what makes levels hard in these cases are precisely the deadly zones.

Several obvious open problems arise from these results. First, it would be interesting to identify more sub-games that are solvable in {\P}\@. Two possible candidates would be the sub-games with only Bombers or only Blockers, and no deadly zones. Also, we conjecture that stronger inapproximability results may be obtained, perhaps for levels with more than one available skill type. Finally, the parameterized complexity of Lemmings may be studied, for instance when the parameter is the number of available skills of a specific type, or the height or width of a level, etc.

\end{document}